\documentclass[a4paper]{llncs}

\usepackage{multicol}
\usepackage[english]{babel}
\usepackage{wrapfig}

\usepackage{color}
\usepackage{xcolor}

\usepackage{amsmath}
\usepackage{amssymb}
\usepackage{caption}
\usepackage{subcaption}
\RequirePackage[noline,noend,linesnumbered]{algorithm2e}

\usepackage{multirow}

\newcommand\dint[3]{\Delta_{#1}^{#2,#3}}

\include{mathdefs}

\auxfun{prepare}
\auxfun{effect}
\auxfun{eval}
\auxfun{apply}
\auxfun{merge}
\auxfun{inc}
\auxfun{dec}
\auxfun{rd}
\auxfun{add}
\auxfun{rmv}
\auxfun{elements}
\let\max=\relax
\auxfun{max}

\usepackage{xspace}
\newcommand\dsrdt{\hbox{$\delta$-CRDT}\xspace}

\auxfun{size}

\usepackage{soul}

\begin{document}

\title{Efficient State-based CRDTs by Delta-Mutation}
\author{Paulo S\'{e}rgio Almeida \and Ali Shoker\and Carlos Baquero}
\institute{HASLab/INESC TEC and Universidade do Minho, Portugal}

\maketitle

\begin{abstract}
CRDTs are distributed data types that make eventual consistency of a
distributed object possible and non ad-hoc. Specifically, state-based CRDTs
ensure convergence  
through disseminating the entire state, that may be large, and merging it to other
replicas; whereas operation-based CRDTs disseminate operations (i.e., small states)
assuming an exactly-once reliable dissemination layer.
We introduce \emph{Delta State Conflict-Free
Replicated Datatypes} (\dsrdt) 
that can achieve the best of both worlds: 
small messages with an incremental nature,
as in operation-based CRDTs, 
disseminated over unreliable communication channels, 
as in traditional state-based CRDTs.  
This is achieved by defining \emph{$\delta$-mutators} to return a 
\emph{delta-state}, typically with a 
much smaller size than the full state, that is joined to both: local and remote states.
We introduce the \dsrdt framework, and we explain it through establishing 
a correspondence to current state-based CRDTs. 
In addition, we present an anti-entropy algorithm 
that ensures causal consistency,  
and we introduce two \dsrdt specifications of well-known replicated datatypes.

\end{abstract}

\section{Introduction}
\label{sec:intro}

Eventual consistency (EC) is a relaxed consistency model that is often adopted
by large-scale distributed
systems~\cite{riak+crdt,syn:optim:rep:1433,app:rep:optim:1606}
where availability must be maintained, despite outages and partitioning,
whereas delayed consistency is acceptable. 
A typical approach in EC systems is to allow replicas of a
distributed object to temporarily diverge, provided that they can eventually
be reconciled into a common state. To avoid application-specific
reconciliation methods, costly and error-prone, \emph{Conflict-Free Replicated
Data Types}~(CRDTs)~\cite{rep:syn:sh138,syn:rep:sh143} were introduced,
allowing the design of self-contained distributed data types that are always
available and eventually converge when all operations are
reflected at all replicas. Though CRDTs are being deployed in
practice~\cite{riak+crdt}, more work is still required to improve their
design and performance.

CRDTs support two complementary designs: \emph{operation-based} (or op-based) and \emph{state-based}.
In op-based designs~\cite{alg:rep:sh132,syn:rep:sh143}, the execution of an operation is
done in two phases: \emph{prepare} and \emph{effect}. The former is
performed only on the local replica and looks at the operation and current
state to produce a message that aims to represent the operation, which is then
shipped to all replicas.  Once received, the representation of the
operation is applied remotely using \emph{effect}. 
On the other hand, in a state-based
design~\cite{app:1639,syn:rep:sh143} an operation is only executed on the
local replica state. A replica periodically propagates its local changes to
other replicas through shipping its entire state. A received state is
incorporated with the local state via a \emph{merge} function that
deterministically reconciles both states. To maintain convergence,
\emph{merge} is defined as a \emph{join}: a least upper bound over a
join-semilattice~\cite{app:1639,syn:rep:sh143}.

Op-based CRDTs have more advantages as they can allow for simpler 
implementations, concise replica state, and smaller messages; however, 
they are subject to some limitations: First, they assume a message dissemination
layer that guarantees reliable exactly-once causal broadcast
(required to ensure idempotence); these guarantees are hard to
maintain since large logs must be retained to prevent duplication even if TCP 
is used~\cite{idempotenceHelland:2012}.
Second, 
membership management is a hard task in op-based systems especially once the number 
of nodes gets larger or due to churn problems, since 
all nodes must be coordinated by the middleware. Third, the op-based approach requires 
operations to be executed individually (even when batched) on all nodes.

The alternative is to use state-based systems which are deprived from these 
limitations. 
However, a major drawback in current state-based CRDTs is the communication overhead of
shipping the entire state, which can get very large in size. For instance, the
state size of a \emph{counter} CRDT (a vector of integer counters, one per
replica) increases with the number of replicas; whereas in a \emph{grow-only
Set}, the state size depends on the set size, that grows as more operations are
invoked.
This communication overhead limits the use of state-based CRDTs to data-types
with small state size (e.g., counters are reasonable while sets are not).
Recently, there has been a demand for CRDTs with large state sizes (e.g., in
RIAK DT Maps~\cite{Brown:2014:RDM:2596631.2596633} that can compose multiple CRDTs).

In this paper, we rethink the way state-based CRDTs should be designed, having
in mind the problematic shipping of the entire state. Our aim is to ship a
\emph{representation of the effect} of recent update operations on the state, rather
than the whole state, while preserving the idempotent nature of \emph{join}.
This ensures convergence over unreliable communication (on the contrary to op-based).
To achieve this, we introduce \emph{Delta State-based CRDTs} (\dsrdt): a state
is a join-semilattice that results from the join of multiple fine-grained
states, i.e., \emph{deltas}, generated by what we call
\emph{$\delta$-mutators}. \emph{$\delta$-mutators} are new versions of the datatype mutators that
return the effect of these mutators on the state. In this way, deltas can be
temporarily retained in a buffer to be shipped individually (or joined in groups) instead
of shipping the entire object. The changes to the local state are then
incorporated at other replicas by joining the shipped deltas with their own
states. 

The use of ``deltas'' (i.e., incremental states) 
may look intuitive in state dissemination; 
however, this is not the case for state-based CRDTs. 
The reason is that once a node receives an entire 
state, merging it locally is simple since there is no 
need to care about causality, as both states are self-contained (including meta-data).
The challenge in \dsrdt is that individual deltas are now ``state fragments'' and 
must be causally merged to maintain the correct semantics. 
This raises the following questions: 
is merging deltas semantically equivalent to merging
entire states in CRDTs? If not, what are the sufficient conditions to make this true 
in general? And under what constraints causal consistency is maintained? 
This paper answers these questions and presents corresponding proofs and examples.

We address the challenge of designing a new \dsrdt that conserves the correctness properties and semantics 
of an existing CRDT by  
 establishing a relation between the novel
$\delta$-mutators with the original CRDT mutators. 
We then show how to ensure causal consistency using deltas through introducing the concept 
of \emph{delta-interval} and the
\emph{causal delta-merging condition}. Based on these, we then present an
anti-entropy algorithm for \dsrdt, where sending and then joining
delta-intervals into another replica state produces the same effect as if the
entire state had been shipped and joined. 

As the area of CRDTs is relatively new, we illustrate our approach by
explaining a simple $counter$ $\dsrdt$
specification; then we introduce a challenging non-trivial specification
for a widely used datatype: Optimized Add-Wins Observed-Remove 
Sets~\cite{rep:opt:sh151}; and finally we present a novel design for 
an Optimized Multi-Value Register with meta-data reduction. 
In addition, we make a basic $\dsrdt$ C++ library available
online~\cite{deltaCode} for various CRDTs: 
GSet, 2PSet, GCounter, PNCounter, AWORSet, RWORSet, MVRegister, LWWSet, etc. 
Our experience shows that a
$\dsrdt$ version can be devised for most CRDTs, however, this requires some
design effort that varies with the complexity of different CRDTs. This is referred
to the ad-hoc way CRDTs are designed in general (which is also required
for $\dsrdt$s). To the best of our knowledge, no model has been 
introduced so far to make designing CRDTs generic instead of being type-specific.

\section{System Model}

Consider a distributed system with nodes containing local memory, with no
shared memory between them.
Any node can send messages to any other node. The network is asynchronous;
there is no global clock, no bound on the time a message takes to
arrive, and no bounds are set on relative processing speeds. The network is unreliable:
messages can be lost, duplicated or reordered (but are not corrupted). Some
messages will, however, eventually get through: if a node sends infinitely
many messages to another node, infinitely many of these will be delivered. In
particular, this means that there can be arbitrarily long partitions, but
these will eventually heal.
Nodes have access to durable storage; they can crash but will eventually 
recover with the content of the durable storage just before crash the occurred.
Durable state is written atomically at each state transition.
Each node has access to its globally unique identifier in a set $\ids$.

\section{A Background of State-based CRDTs}
\label{sec:back}
\emph{Conflict-Free Replicated Data Types}~\cite{rep:syn:sh138,syn:rep:sh143} (CRDTs) are distributed datatypes that allow different replicas of a distributed CRDT instance to diverge and ensures that, eventually, all replicas converge to the same state. State-based CRDTs achieve this through propagating updates of the local state by disseminating the entire state across replicas. The received states are then merged to remote states, leading to convergence (i.e., consistent states on all replicas).

A state-based CRDT consists of a triple $(S, M, Q)$, where $S$ is a
join-semi\-lattice~\cite{lattices:Brian:2002}, $Q$ is a set of query functions (which return some result without modifying the state), and $M$ is a set of mutators that perform updates; a mutator $m \in M$ takes a state $X \in S$ as input and returns a new state $X' = m(X)$. A join-semilattice is a set with a \emph{partial
order} $\pleq$ and a binary \emph{join} operation $\join$ that returns the
\emph{least upper bound }(LUB) of two elements in $S$; a \emph{join} is designed to be commutative, associative, and idempotent. Mutators are defined in such a way to be \emph{inflations}, i.e., for any mutator $m$ and state $X$, the following holds:
\[ X \pleq m(X) \]
In this way, for each replica there is a monotonic sequence of states, defined under the lattice partial order, where each subsequent state subsumes the previous state when joined elsewhere.

Both query and mutator operations are always available since they are performed
using the local state without requiring inter-replica communication; however,
as mutators are concurrently applied at distinct replicas, replica states will
likely diverge. Eventual convergence is then obtained using an
\emph{anti-entropy} protocol that periodically ships the entire local state to
other replicas.  Each replica merges the received state with its local state
using the \emph{join} operation in $S$.  Given the mathematical properties of
\emph{join}, if mutators stop being issued, all replicas eventually converge to
the same state. i.e. the least upper-bound of all states involved.  State-based
CRDTs are interesting as they demand little guarantees from the dissemination
layer, working under message loss, duplication, reordering, and temporary
network partitioning, without impacting availability and eventual convergence.

\begin{wrapfigure}{r}{0.5\textwidth}
\vspace*{-4\baselineskip}
  \begin{center}
\begin{eqnarray*}
\Sigma & = & \ids \map \nat \\
\sigma^0_i & = & \{ \} \\
\inc_i(m) & = &  m\{i \mapsto m(i)+1\}\\
\af{value}_i(m) & = & \sum_{i \in \ids} m(i) \\
m \join m' & = & \{ (i, \max(m(i),m'(i))) | i \in \ids \} \\
\end{eqnarray*}
\caption{State-based Counter CRDT; replica $i$.}
\label{fig:statectr}
\end{center}
\end{wrapfigure}

\textbf{Example.} Fig.~\ref{fig:statectr} represents a state-based increment-only counter. The
CRDT state $\Sigma$ is a map from replica identifiers to positive integers.
Initially, $\sigma^0_i$ is an empty map (assuming that unmapped keys
implicitly map to zero, and only non zero mappings are stored). A single 
mutator, i.e., $\inc$, is defined that increments the value corresponding to
the local replica $i$ (returning the updated map). The query operation
$\af{value}$ returns the counter value by adding the integers in the map
entries.  The join of two states is the point-wise maximum of the maps.

\textbf{Weaknesses.} The main weakness of state-based CRDTs is the cost of dissemination of updates,
as the full state is sent. In this simple example of counters, even though
increments only update the value corresponding to the local replica $i$, the
whole map will always be sent in messages though the other map values remained
intact (since no messages have been received and merged).

It would be interesting to only ship the recent modification incurred on the
state. This is, however, not possible with the current model of state-based
CRDTs as mutators always return a full state. Approaches which simply ship
operations (e.g., an ``increment $n$'' message), like in operation-based
CRDTs, require reliable communication (e.g., because increment is not
idempotent).
In contrast, our approach allows producing and encoding recent mutations in an
incremental way, while keeping the advantages of the state-based approach,
namely the idempotent, associative, and commutative properties of join.

\section{Delta-state CRDTs}
\label{sec:model}

We introduce \emph{Delta-State Conflict-Free Replicated Data Types}, or \dsrdt
for short, as a new kind of state-based CRDTs, in which \emph{delta-mutators}
are defined to return a \emph{delta-state}: a value in the same
join-semilattice which represents the updates induced by the mutator on the
current state.

\begin{definition}[Delta-mutator]
A delta-mutator $m^\delta$ is a function, corresponding to an update operation, which takes a state $X$ in a
join-semilattice $S$  as parameter and returns a delta-mutation $m^\delta(X)$, also in $S$.
\end{definition}

\begin{definition}[Delta-group]
A delta-group is inductively defined as either a delta-mutation or a join of
several delta-groups.
\end{definition}

\begin{definition}[\dsrdt]
A \dsrdt consists of a triple $(S, M^\delta, Q)$, where $S$ is a
join-semilattice, $M^\delta$ is a set of delta-mutators, and $Q$ a set of
query functions, where the state transition at each replica is given by either
joining the current state $X \in S$ with a delta-mutation:

\[
X' = X \join m^\delta(X),
\]
or joining the current state with some received delta-group $D$:
\[
X' = X \join D.
\]
\end{definition}

In a \dsrdt, the effect of applying a mutation, represented by a
delta-mutation $\delta = m^\delta(X)$, is decoupled from the resulting state
$X' = X \join \delta$, which allows shipping this $\delta$ rather than the
entire resulting state $X'$.  All state transitions in a $\dsrdt$, even upon
applying mutations locally, are the result of some join with the current state.
Unlike standard CRDT mutators, delta-mutators do not need to be inflations in order to inflate a state; this is however ensured by joining their output, i.e., deltas, into the current state.

In principle, a delta could be shipped immediately to remote replicas once
applied locally.
For efficiency reasons, multiple deltas returned by applying several
delta-mutators can be joined locally into a delta-group and retained in
a buffer. The delta-group can then be shipped to remote replicas to be joined
with their local states. Received delta-groups can optionally be joined into
their buffered delta-group, allowing transitive propagation of
deltas. A full state can be seen as a special (extreme) case of a delta-group.

If the causal order of operations is not important and the intended aim is
merely eventual convergence of states, then delta-groups can be shipped using
an unreliable dissemination layer that may drop, reorder, or duplicate
messages.  Delta-groups can always be re-transmitted and re-joined, possibly
out of order, or can simply be subsumed by a less frequent sending of the full
state, 
 e.g. for performance reasons or when
doing state transfers to new members. Due to space limits, we only address
causal consistency in this paper, while information about state convergence 
can be found in the associated technical report~\cite{delta:almeida:2014}.

\subsection{Delta-state decomposition of standard CRDTs}

A \dsrdt $(S,M^\delta,Q)$ is a \emph{delta-state decomposition} of a 
state-based CRDT $(S,M,Q)$, if for every mutator $m \in M$, we have a
corresponding mutator $m^\delta \in M^\delta$ such that, for every state $X \in
S$:

\[
m(X) = X \join m^\delta(X)
\]

This equation states that applying a delta-mutator and joining into the
current state should produce the same state transition as applying the
corresponding mutator of the standard CRDT.

Given an existing state-based CRDT (which is always a trivial decomposition of
itself, i.e., $m(X) = X \join m(X)$, as mutators are inflations), it will be
useful to find a non-trivial decomposition  such that delta-states returned
by delta-mutators in $M^\delta$ are smaller than the resulting 
state: \[ \size(m^\delta(X)) \ll \size(m(X)) \]

\subsection{Example: \dsrdt Counter}
\label{sec:counter}

\begin{wrapfigure}[14]{r}{0.5\textwidth}
\vspace*{-2\baselineskip}
\begin{center} \begin{eqnarray*}
\Sigma & = & \ids \map \nat \\ \sigma^0_i & = & \{ \} \\ \inc_i^\delta(m) & = &
\{i \mapsto m(i)+1\} \\ \af{value}_i(m) & = & \sum_{i \in \ids} m(i) \\ m \join m'
& = & \{ (i, \max(m(i),m'(i))) | i \in \ids \} \\ \end{eqnarray*}
\caption{A \dsrdt counter; replica $i$.} \label{fig:deltactr}
\end{center} 
\end{wrapfigure}

Fig.~\ref{fig:deltactr} depicts a \dsrdt
specification of a counter datatype that is a delta-state decomposition of the
state-based counter in Fig.~\ref{fig:statectr}.  The state, join and
$\af{value}$ query operation remain as before.  Only the mutator $\inc^\delta$
is newly defined, which increments the map entry corresponding to the local
replica and only returns that entry, instead of the full map as $\inc$ in the
state-based CRDT counter does.  This maintains the original semantics of the
counter while allowing the smaller deltas returned by the delta-mutator to be
sent, instead of the full map.  As before, the received payload (whether one or
more deltas) might not include entries for all keys in $\ids$, which are assumed
to have zero values.  The decomposition is easy to understand in this example since the equation $\inc_i(X) = X \join \inc_i^\delta(X)$ holds as $m\{i \mapsto m(i)+1\}= m \join \{i \mapsto m(i)+1\}$. In other words, the
single value for key $i$ in the delta, corresponding to the local replica
identifier, will overwrite the corresponding one in $m$ since the former maps to a higher value (i.e., using $\max$).
Here it can be noticed that: (1) a delta \emph{is} just a state, that can be joined possibly several times without requiring exactly-once delivery, and without being a
representation of the ``increment'' operation {(as in operation-based CRDTs), which is itself non-idempotent; (2) joining deltas into a delta-group and disseminating delta-groups at a lower rate than the operation rate reduces data communication overhead, since multiple increments from a given source can be collapsed into a single state counter.

\section{State Convergence}
\label{sec:conv}

In the \dsrdt execution model, and regardless of the anti-entropy
algorithm used, a replica state always evolves by joining the current
state with some \emph{delta}: either the result of a delta-mutation, or some
arbitrary delta-group (which itself can be expressed as a join of
delta-mutations).  Therefore, all states can be expressed as joins of
delta-mutations, which makes state convergence in \dsrdt easy to achieve: it is
enough that all delta-mutations generated in the system reach every replica,
as expressed by the following proposition.

\begin{proposition}(\dsrdt convergence)
\label{prop:convergence}
Consider a set of replicas of a \dsrdt object, replica $i$ evolving along a
sequence of states $X_i^0=\bot, X_i^1, \ldots$, each replica performing
delta-mutations of the form $m_{i,k}^\delta(X_i^k)$ at some subset of its
sequence of states, and evolving by joining the current state either
with self-generated deltas or with delta-groups received from others. If each
delta-mutation $m_{i,k}^\delta(X_i^k)$ produced at each replica is joined
(directly or as part of a delta-group) at least once with every
other replica, all replica states become equal.
\end{proposition}

\begin{proof}
Trivial, given the associativity, commutativity, and idempotence of the join
operation in any join-semilattice.
\end{proof}

This opens up the possibility of having anti-entropy algorithms that are only devoted to enforce convergence, without necessarily providing causal consistency (enforced in standard CRDTs); thus, making a trade-off between performance and consistency guarantees.  For
instance, in a counter (e.g., for the number of \emph{likes} on a social
network), it may not be critical to have causal consistency, but merely not to
lose increments and achieve convergence.

\subsection{Basic Anti-Entropy Algorithm}
A basic anti-entropy algorithm that ensures eventual convergence in \dsrdt is presented in Algorithm~\ref{alg:basic-algo}. For the node corresponding to replica $i$, the durable state, which persists after a
crash, is simply the \dsrdt state $X_i$. The volatile state $D$ stores a
delta-group that is used to accumulate deltas before eventually sending it to other
replicas.  Without loss of generality, we assume that the join-semilattice has a
bottom $\bot$, which is the initial value for both $X_i$ and $D_i$.

{
\auxfun{send}
\auxfun{receive}
\auxfun{operation}
\auxfun{ack}
\auxfun{random}

\begin{algorithm}[t]

\begin{multicols}{2}
\DontPrintSemicolon
\SetKwBlock{inputs}{inputs:}{}
\SetKwBlock{dstate}{durable state:}{}
\SetKwBlock{vstate}{volatile state:}{}
\SetKwBlock{periodically}{periodically}{}
\SetKwBlock{on}{on}{}

\inputs{
  $n_i \in \pow\ids$, set of neighbors \;
  $t_i \in \bool$, true for transitive mode\;
  $\af{choose}_i \in S \times S \to S$, ship state or delta\;
}

\dstate{
  $X_i \in S$, CRDT state; initially $X_i = \bot$ \;
}

\vstate{
  $D_i \in S$, join of deltas; initially $D_i = \bot$ \;
}

\on({$\operation_i(m^\delta)$}){
  $d = m^\delta(X_i)$ \;
  $X_i' = X_i \join d$ \;
  $D_i' = D_i \join d$ \;
}

\BlankLine
\periodically(){
  $m = \af{choose}_i(X_i, D_i)$ \;
  \For{$j \in n_i$}{
    $\send_{i,j}(m)$ \;
  }
  $D_i' = \bot$ \;
}
\BlankLine
\on({$\receive_{j,i}(d)$}){
  $X_i' = X_i \join d$ \;
  \uIf{$t_i$}{ $D_i' = D_i \join d$ \; }
  \uElse { $D_i' = D_i$ \; }
}
\end{multicols}
\bigskip
\caption{Basic anti-entropy algorithm for \dsrdt. \label{alg:basic-algo}}
\end{algorithm}
}

When an operation is performed, the corresponding delta-mutator $m^\delta$ is
applied to the current state of $X_i$, generating a delta $d$. This delta is joined
both with $X_i$ to produce a new state, and with $D$.
In the same spirit of standard state based CRDTs, a node sends its messages in a periodic fashion, where the message payload is either the delta-group $D_i$ or the full state $X_i$; this decision is made by the function $\af{choose}_i$ which returns one of them. To keep the algorithm simple, a node simply broadcasts its messages without distinguishing between neighbors. After each send, the delta-group is reset to $\bot$.

Once a message is received, the payload $d$ is joined into the current \dsrdt
state.  The basic algorithm operates in two modes: (1) a \emph{transitive} mode (when $t_i$ is true) in which $m$ is also joined into $D$, allowing transitive propagation of delta-mutations; meaning that, deltas received at node $i$ from some node $j$ can later be sent to some other node $k$; (2) a
\emph{direct} mode where a delta-group is exclusively the join of local
delta-mutations ($j$ must send its deltas directly to $k$).
The decisions of whether to send a delta-group versus the full state
(typically less periodically), and whether to use the transitive or direct
mode are out of the scope of this paper. In general, decisions can be made considering many criteria like delta-groups size, state size, message loss distribution assumptions, and network topology.

\section{Causal Consistency}
\label{sec:causal}

Traditional state-based CRDTs converge using joins of the full state, which
implicitly ensures per-object causal consistency~\cite{DBLP:conf/popl/BurckhardtGYZ14}: each state of
some replica of an object reflects the causal past of operations on the object
(either applied locally, or applied at other replicas and transitively
joined).

Therefore, it is desirable to have \dsrdt{}s offer the same causal-consistency guarantees that standard state-based CRDTs offer. This raises the question about how can delta propagation and merging of \dsrdt be constrained (and expressed in an anti-entropy algorithm) in such a manner to give the same results as if a standard state-based CRDT was used. Towards this objective, it is useful to define a particular kind of delta-group, which we call a \emph{delta-interval}:

 \begin{definition}[Delta-interval]
Given a replica $i$ progressing along the states
$X_i^0, X_ i^1 , \ldots$, by joining delta $d_i^k$ (either local
delta-mutation or received delta-group) into $X_i^k$ to obtain
$X_i^{k+1}$, a delta-interval $\dint iab$ is a delta-group resulting from
joining deltas $d_i^a, \ldots, d_i^{b-1}$:
 \[
\dint iab = \bigjoin \{ d_i^k | a \leq k < b \}
 \]
 \end{definition}

The use of delta-intervals in anti-entropy algorithms will be a key ingredient
towards achieving causal consistency. We now define a restricted kind of
anti-entropy algorithm for \dsrdt{}s.

\begin{definition}[Delta-interval-based anti-entropy algorithm]
A given anti-entropy algorithm for \dsrdt{}s is delta-interval-based, if all
deltas sent to other replicas are
delta-intervals.
\end{definition}

Moreover, to achieve causal consistency the next condition must satisfied:

\begin{definition}[Causal delta-merging condition]
\label{def:c_cond}
A delta-interval based anti-entropy algorithm is said to satisfy the causal
delta-merging condition if the algorithm only joins $\dint jab$ from replica
$j$ into state $X_i$ of replica $i$ that satisfy:
\[
X_i \pgeq X_j^a.
\]
\end{definition}

This means that a delta-interval is
only joined into states that at least reflect (i.e., subsume) the state into which the first delta in the interval was previously joined.
The causal delta-merging condition is important since any delta-interval
based anti-entropy algorithm of a \dsrdt that satisfies it, can be used to
obtain the same outcome of standard CRDTs; this is formally stated in Proposition~\ref{prop:corr}.

\begin{proposition}(CRDT and \dsrdt correspondence)
\label{prop:corr}
Let $(S,M,Q)$ be a standard state-based CRDT and $(S,M^\delta,Q)$ a corresponding delta-state
decomposition. Any \dsrdt state reachable by an execution $E^\delta$
over $(S,M^\delta,Q)$, by a delta-interval based anti-entropy algorithm $A^\delta$
satisfying the causal delta-merging condition, is equal to a state
resulting from an execution $E$ over $(S,M,Q)$, having the corresponding data-type
operations, by an anti-entropy algorithm $A$ for state-based CRDTs.
\end{proposition}

\begin{proof}
Please see the associated technical report~\cite{delta:almeida:2014}.
\end{proof}

\begin{corollary}(\dsrdt causal consistency)
\label{prop:causal-consistency}
Any \dsrdt in which states are propagated and joined using a delta-interval-based anti-entropy algorithm satisfying the causal delta-merging condition ensures causal consistency.
\end{corollary}

\begin{proof}
From Proposition~\ref{prop:corr} and causal consistency of state-based CRDTs. 
\end{proof}

\subsection{Anti-Entropy Algorithm for Causal Consistency}
Algorithm~\ref{alg:causal-algo} is a delta-interval based anti-entropy
algorithm which enforces the causal delta-merging condition. It can be used
whenever the causal consistency guarantees of standard state-based CRDTs are
needed. For simplicity, it excludes some optimizations that are
important, but easy to derive, in practice. The algorithm distinguishes
neighbor nodes, and only sends them delta-intervals that are joined at the
receiving node, obeying the delta-merging condition.

Each node $i$ keeps a contiguous sequence of deltas $d_i^l, \ldots, d_i^u$ in a
map $D$ from integers to deltas, with $l = \min(\dom(D))$ and $u =
\max(\dom(D))$. The sequence numbers of deltas are obtained from the counter
$c_i$ that is incremented when a delta (whether a delta-mutation or
delta-interval received) is joined with the current state.  Each node $i$
keeps an acknowledgments map $A$ that stores, for each neighbor $j$, the
largest index $b$ for all delta-intervals $\dint iab$ acknowledged by $j$
(after $j$ receives $\dint iab$ from $i$ and joins it into $X_j$).

Node $i$ sends a delta-interval $d=\dint iab$ with a $(\af{delta}, d,
b)$ message; the receiving node $j$, after joining $\dint iab$ into its replica
state, replies with an acknowledgment message $(\af{ack}, b)$; if an ack
from $j$ was successfully received by node $i$, it updates the entry of $j$ in the acknowledgment map, using the $\max$ function. This handles possible old duplicates and messages arriving out of order.

{
\auxfun{send}
\auxfun{receive}
\auxfun{operation}
\auxfun{ack}
\auxfun{random}

\begin{algorithm}[t]
\begin{multicols}{2}
\DontPrintSemicolon
\SetKwBlock{inputs}{inputs:}{}
\SetKwBlock{dstate}{durable state:}{}
\SetKwBlock{vstate}{volatile state:}{}
\SetKwBlock{periodically}{periodically}{}
\SetKwBlock{on}{on}{}

\inputs{
  $n_i \in \pow\ids$, set of neighbors \;
}

\dstate{
  $X_i \in S$, CRDT state; initially $X_i = \bot$ \;
  $c_i \in \nat$, sequence number; initially $c_i = 0$ \;
}

\vstate{
  $D_i \in \nat \map S$, sequence of deltas; initially $D_i = \{\}$ \;
  $A_i \in \ids \map \nat$, acknowledges map; initially $A_i = \{\}$ \;
}

\on({$\receive_{j,i}(\af{delta}, d, n)$}){
  \uIf{$d \not \pleq X_i$}{
    $X_i' = X_i \join d$ \;
    $D_i' = D_i \{ c_i \mapsto d \}$ \;
    $c_i' = c_i + 1$ \;
  }
  $\send_{i,j}(\ack, n)$ \;
}

\on({$\receive_{j,i}(\ack, n)$}){
  $A_i' = A_i \{ j \mapsto \max(A_i(j), n)\}$ \;
}

\on({$\operation_i(m^\delta)$}){
  $d = m^\delta(X_i)$ \;
  $X_i' = X_i \join d$ \;
  $D_i' = D_i \{ c_i \mapsto d \}$ \;
  $c_i' = c_i + 1$ \;
}

\periodically( // ship delta-interval or state){
  $j = \random(n_i)$ \;
  \uIf{$D_i = \{\} \lor \min(\dom(D_i)) > A_i(j)$}{
    $d = X_i$
  } \uElse{
    $d = \bigjoin\{ D_i(l) | A_i(j) \leq l < c_i \}$ \;
  }
  \uIf{$A_i(j) < c_i$}{
    $\send_{i,j}(\af{delta}, d, c_i)$ \;
  }
}

\periodically( // garbage collect deltas){
  $l = \min \{ n | (\_, n) \in A_i \}$ \;
  $D_i' = \{ (n,d) \in D_i | n \geq l \}$ \;
}
\BlankLine
\end{multicols}
\bigskip
\caption{Anti-entropy algorithm ensuring causal consistency of \dsrdt.}
\label{alg:causal-algo}

\end{algorithm}
}

Like the \dsrdt state, the counter $c_i$ is also kept in a durable storage.
This is essential to avoid conflicts after potential crash and recovery
incidents.  Otherwise, there would be the danger of receiving some delayed ack,
for a delta-interval sent before crashing, causing the node to skip sending
some deltas generated after recovery, thus violating the delta-merging
condition.

The algorithm for node $i$ periodically picks a random neighbor $j$. In
principle, $i$ sends the join of all deltas starting from the latest delta
acked by $j$ and forward. Exceptionally, $i$ sends the entire state in two
cases: (1) if the sequence of deltas $D_i$ is empty, or (2) if $j$ is
expecting from $i$ a delta that was already removed from $D_i$ (e.g., after a
crash and recovery, when both deltas and the ack map, being volatile state,
are lost); $i$ tracks this in $A_i(j)$.  To garbage collect old deltas, the
algorithm periodically removes the deltas that have been acked by \emph{all}
neighbors.

\begin{proposition}
Algorithm~\ref{alg:causal-algo} produces the same reachable states as a
standard algorithm over a CRDT for which the \dsrdt is a decomposition.
\label{prop:causal-algo}
\end{proposition}

\begin{proof}
Please see the associated technical report~\cite{delta:almeida:2014}.
\end{proof}

\section{\dsrdt{}s for Add-Wins OR-Sets}

An Add-wins Observed-Remove Set (OR-set) is a well-known CRDT datatype that offers the
same sequential semantics of a sequential set and adopts a specific
resolution semantics for operations that concurrently add and remove the same
element. Add-wins means that an add prevails over a concurrent remove. Remove
operations, however, only affect elements added by causally preceding adds. The 
purpose of these \dsrdt OR-set versions is to design $\delta$-mutators that return 
small deltas to be lightly disseminated, as discussed above, instead of shipping 
the entire state as in classical CRDTs~\cite{rep:syn:sh138,syn:rep:sh143,rep:opt:sh151}.

\subsection{Add-wins OR-Set with tombstones}
Fig.~\ref{subf:orset} depicts a simple, but inefficient, \dsrdt implementation
of a state-based add-wins OR-Set. The state $\Sigma$ consists of a set of
tagged elements and a set of tags, acting as tombstones. Globally unique tags
of the form $\ids \times \nat$ are used and ensured by pairing a replica 
identifier in $\ids$ with a monotonically increasing natural counter. 
Once an element $e \in E$ is added to the set, the delta-mutator $\add^\delta$
creates a globally unique tag by incrementing the highest tag present in its 
local state and that was created by replica $i$ itself ($\max$ returns 0 if 
no tag is present). This tag is paired with value $e$ and stored as a new 
unique triple in $s$. Since an ``add'' wins any concurrent ``remove'', 
removing an element $e$ should only be tombstoned if it was 
preceded by an add operation (i.e., the element is in $s$), otherwise it has
no effect. Consequently, the delta-mutator $\rmv^\delta$ retains in the tombstone set
all tags associated to element $e$, being removed from the local state. This is essential 
to prevent a removed element to reappear once the local state is merged with 
another replica state that still have that element (i.e., it has not been removed yet 
remotely as replicas are loosely coupled).
The function $\elements$ returns only the elements that are added but not yet tombstoned.
Join $\join$ simply unions the respective sets that are, therefore,
both grow-only.

\begin{figure}[t]
\scriptsize
\hspace*{-5mm}
\begin{subfigure}[b]{0.5\textwidth}
\begin{eqnarray*}
\Sigma & = & \mathcal{P}(\ids \times \nat \times E) \times \mathcal{P}(\ids \times \nat) \\ 
\sigma_i^0 & = & (\{\},\{\}) \\
\add_i^\delta(e,(s,t)) & = &
(\{(i,n+1,e)\},\{\}) \\
&& \textrm{with } n = \max(\{k | (i,k,\_) \in s\}) \\
\rmv_i^\delta(e,(s,t)) & = & (\{\},\{(j,n) | (j,n,e) \in s\})\\
\elements_i((s,t)) & = & \{e | (j,n,e) \in s \land (j,n) \not\in t\}\\
(s,t) \join (s',t') & = & (s \cup s', t \cup t') \\
&&
\end{eqnarray*}

\caption{With Tombstones}
\label{subf:orset}
\end{subfigure}
~
\begin{subfigure}[b]{0.5\textwidth}
\begin{eqnarray*}
\Sigma & = & \mathcal{P}(\ids \times \nat \times E) \times \mathcal{P}(\ids \times \nat) \\ 
\sigma_i^0 & = & (\{\},\{\}) \\
\add_i^\delta(e,(s,c)) & = &
(\{(i,n+1,e)\}, \{(i,n+1)\}) \\
&& \textrm{with } n = \max(\{k | (i,k) \in c\}) \\
\rmv_i^\delta(e,(s,c)) & = & (\{\},\{(j,n) | (j,n,e) \in s\}) \\
\elements_i((s,c)) & = & \{e | (j,n,e) \in s\}\\
(s,c) \join (s',c') & = & ((s \cap s') \cup
\{ (i,n,e) \in s | (i,n) \not\in c' \} \\
&& \cup
\{ (i,n,e) \in s' | (i,n) \not\in c \}
, c \cup c')
\end{eqnarray*}
\caption{Without Tombstones (optimized)}
\label{subf:optorset}
\end{subfigure}
\caption{Add-wins observed-remove \dsrdt set, replica $i$.}
\label{fig:orset}
\end{figure}

\subsection{Optimized Add-wins OR-Set}
A more efficient design is presented in Fig.~\ref{subf:optorset} allowing 
also the set of tagged elements (i.e., tombstone set above) to shrink as elements are removed. 
This design offers the same semantics and have a similar state structure to the former;
however, it has a different behavior. Now, $\elements$ returns all the elements in 
the tagged set $s$, without consulting $t$ as before. 
Added and removed items are now tagged in the \emph{causal context set} $c$.
Although, the set $c$ and $t$ look similar in structure, they have a different behavior
(we call it $c$ instead of $t$ to remove this confusion): a tombstone set $t$
simply stores all removed elements tags, while $c$ retains only the causal information needed to 
add/remove an element. For presentation simplicity, $c$ in Fig.~\ref{subf:optorset} simply retains
all removed elements tags; however, after compression, $c$ will be very concise and look 
different from $t$; this is explained in the next section.

Adding an element creates a unique tag by resorting to the causal context $c$ (instead of $s$). 
The tag is paired with the element and added to $s$ (as before). 
The difference is that the new tag is also added to the causal context set $c$. 
The delta-mutator $\rmv^\delta$ is the same as before, adding all tags associated to the
element being removed to $c$. The desired semantics are maintained by the novel join
operation $\join$. To join two states, their causal contexts $c$ are simply
unioned; whereas, the new element set $s$ only preserves: (1) the triples
present in both sets (therefore, not removed in either), and also (2) any
triple present in one of the sets and whose tag is not present in the causal
context of the other state.

\subsubsection{Causal Context Compression}
In practice, the causal context $c$ can be efficiently compressed without any loss
of information. When using an anti-entropy algorithm that provides causal
consistency, e.g., Algorithm~\ref{alg:causal-algo}, then for each replica state $X_i=(s_i, c_i)$ and replica identifier $j
\in \ids$, we have a contiguous sequence:
\[
1 \leq n \leq \max(\{ k | (j,k) \in c_i\}) \implies (j,n) \in c_i.
\]
Thus, the causal context can always be encoded as a compact version 
vector~\cite{Parker:1983:DMI:1313337.1313753} $\ids \map \nat$ that keeps the maximum sequence number for each replica.
Even under non-causal anti-entropy, compression is still possible by
keeping a version vector that encodes the offset of the contiguous sequence of tags from each replica, together with a set for the non-contiguous tags. As anti-entropy
proceeds, each tag is eventually encoded in the vector, and thus the set remains
typically small.
Compression is less likely for the causal context of delta-groups in transit
or buffered to be sent, but those contexts are only transient and smaller than
those in the actual replica states. Moreover, the same techniques that encode contiguous sequences of tags can also be used for transient context compression~\cite{MGS14}.

\section{Optimized Multi-value Register \dsrdt{}}

\begin{figure}[t]
\scriptsize
\begin{eqnarray*}
\Sigma & = & \mathcal{P}(\ids \times \nat \times V) \times \mathcal{P}(\ids \times \nat) \\ 
\sigma_i^0 & = & (\{\},\{\}) \\
  \mathsf{wr}_i^\delta(v,(s,c)) & = &  (\{(i,n+1,v)\},\{(i,n+1)\} \cup \{(j,m) | (j,m,\_) \in s\}) \textrm{ with } n = \max(\{k | (i,k) \in c\}) \\
\rd_i((s,c)) & = & \{v | (j,n,v) \in s\}\\
(s,c) \join (s',c') & = & ((s \cap s') \cup
\{ (i,n,v) \in s | (i,n) \not\in c' \} 
\cup
\{ (i,n,v) \in s' | (i,n) \not\in c \}
, c \cup c')
\end{eqnarray*}
\caption{Optimized \dsrdt multi-value register, replica $i$.}
\label{fig:mvreg}
\end{figure}
Multi-Value Registers (MVR) are popular constructions in which a read operation 
returns the set of values concurrently written, but not causally overwritten; 
these values are then reduced to a single value by applications~\cite{app:rep:optim:1606}.
Until now, these types have been implemented by assigning a version vector to each written
value~\cite{rep:syn:sh138,DBLP:conf/popl/BurckhardtGYZ14}. In
Figure~\ref{fig:mvreg}, we show that the optimization that was developed for
Sets, can also be used to compactly tag the values in a multi-value register.
On a write operation $\mathsf{wr}$, it is enough to assign a new scalar tag,
from $\ids \times \nat$, using a replica id $i$ and counter to uniquely tag  
the written value $v$. To ensure that values overwritten are
deleted, the produced causal context $c$ lists all tags associated to those values. Since
those values are absent from the payload set $s$ they will be deleted in replicas that
still have them, applying join definition $\join$ (that is in common with
Figure~\ref{subf:optorset}). The causal context compression techniques
defined earlier also apply here.  

\section{Message Complexity}
\label{sec:complexity}

Our delta-based framework, \dsrdt, clearly introduces significant cost
improvements on messaging.  Despite being a generic framework, \dsrdt requires
delta mutators to be
defined per datatype.  This makes the bit-message complexity datatype-based
rather than generic.  To give an intuition about this complexity, we address
the three datatypes introduced above: $counter$, OR-Set, and MVR.

\textbf{Counters.} In classical state-based counter CRDTs, the entire map of the $counter$ is
shipped. As the map-size grows with the number of replicas, this leads a
bit-message complexity of $\widetilde{O}(\left| \ids \right|)$\footnote{$\widetilde{O}$ is a variant of big $O$ ignoring logarithmic factors in the size of integers and ids.}.  Whereas, in
the $\dsrdt$ case, only recently updated map entries $\alpha$ are shipped
yielding a bit-complexity $\widetilde{O}(\alpha)$, where $\alpha \ll
\left|\ids\right|$.

\textbf{OR-set.} Shipping in classical OR-set CRDTs delivers the entire state
which yields a bit-message complexity of $O(S)$, where S is the
state-size.  In \dsrdt, only deltas are shipped, which renders a bit-message
complexity $O(s)$ where $s$ represents the size of the recent
updates occurred since the last shipping.  Clearly, $s \ll S$ since the updates
that occur on a state in a period of time are often much less than the total
number of items. 

\textbf{MVR.} In classical MVR, the worst case state is composed of $\left|
\ids \right|$ concurrently written values, each associated with a $\left| \ids
\right|$ sized version vector.  This makes the bit-message complexity
$\widetilde{O}(\left| \ids \right|^2 )$.  In the novel delta design in
Figure~\ref{fig:mvreg}, no version vector is used, whereas the number of
possible values remain the same (summing up the values set $s$ and meta-data in
$c$), this reduces the bit-message complexity to $\widetilde{O}(\left| \ids \right|)$ as
well as the worst case state complexity.

\section{Related Work}
\label{sec:related}
\subsection{Eventually convergent data types.}

The design of replicated systems that are always available and eventually
converge can be traced back to historical designs in
\cite{app:rep:optim:1501,db:rep:optim:1454}, among others.  More recently, replicated data types that
always eventually converge, both by reliably broadcasting operations (called
operation-based) or gossiping and merging states (called state-based), have
been formalized as
CRDTs~\cite{alg:rep:sh132,app:1639,rep:syn:sh138,syn:rep:sh143}. These are also
closely related to Bloom$^L$ \cite{conway2012logic} and Cloud
Types~\cite{burckhardt2012cloud}.

\subsection{Message size.}

A key feature of \dsrdt is message size reduction and coalescing, using
small-sized deltas. The general old idea of using differences between things,
called ``deltas'' in many contexts, can lead to many designs, depending on how
exactly a delta is defined.  The state-based deltas introduced for
Computational CRDTs ~\cite{navalho2013incremental} require an extra
delta-specific merge (in addition to the standard join) which does not ensure
idempotence.  In~\cite{deltaICDCS13}, an improved synchronization method for
non-optimized OR-set CRDT~\cite{rep:syn:sh138} is presented, where delta
information is propagated; in that paper deltas are a collection of items
(related to update events between synchronizations), manipulated and merged
through a protocol, as opposed to normal states in the semilattice. No generic
framework is defined (that could encompass other data types) and the protocol
requires several communication steps to compute the information to exchange.
Operation-based CRDTs~\cite{rep:syn:sh138,syn:rep:sh143,BAS2014} also support
small message sizes, and in particular, \emph{pure} flavors~\cite{BAS2014} that
restrict messages to the operation name, and possible arguments.  Though pure
operation-based CRDTs allow for compact states and are very fast at the source
(since operations are broadcast without consulting the local state), the model
requires more systems guarantees than \dsrdt do, e.g., exactly-once reliable
delivery and membership information, and impose more complex integration of new
replicas.  The work in \cite{export:211340} shows a different trade-off among
state deltas and pure operations, by tagging operations and creating a globally
stable log of operations while allowing local transient logs to preserve
availability.  While having other advantages, the creation of this global log
requires more coordination than our gossip approach for causally consistent
delta dissemination, and can stall dissemination.

\subsection{Encoding causal histories.}

State-based CRDT are always designed to be causally
consistent~\cite{app:1639,syn:rep:sh143}. Optimized implementations of sets,
maps, and multi-value registers can build on this assumption to keep the
meta-data small~\cite{DBLP:conf/popl/BurckhardtGYZ14}.  In \dsrdt, however,
deltas and delta-groups are normally not causally consistent, and thus the
design of \emph{join}, the meta-data state, as well as the anti-entropy
algorithm used must ensure this. Without causal consistency, the causal context
in \dsrdt can not always be summarized with version vectors, and consequently,
techniques that allow for gaps are often used.  A well known mechanism that
allows for encoding of gaps is found in Concise Version
Vectors~\cite{malkhi2007concise}. Interval Version Vectors~\cite{MGS14}, later
on, introduced an encoding that optimizes sequences and allows gaps, while
preserving efficiency when gaps are absent.

\section{Conclusion}
\label{sec:con}

We introduced the new concept of \dsrdt{}s and devised \emph{delta-mutators} over
state-based datatypes which can detach the changes
that an operation induces on the state. This brings a significant
performance gain as it allows only shipping small states, i.e., \emph{deltas},
instead of the entire state. The significant property in \dsrdt is that it
preserves the crucial properties (idempotence, associativity and
commutativity) of standard state-based CRDT.
 In addition, we have shown how \dsrdt can achieve 

causal consistency; and we presented an anti-entropy algorithm that allows replacing classical
state-based CRDTs by more efficient ones, while preserving their properties.
As an application for our approach, we designed two novel \dsrdt
specifications for two well-known datatypes: an optimized
observed-remove set~\cite{rep:opt:sh151} and an optimized multi-value register~\cite{app:rep:optim:1606}.

Our approach is more relaxed than classical state-based CRDTs, and thus, can
replace them without losing their power since \dsrdt allows shipping delta-states
as well as the entire state. 
Another interesting observation is that \dsrdt can mimic the behavior of operation-based CRDTs, by
shipping individual deltas on the fly but with weaker guarantees from
the dissemination layer.

\bibliographystyle{splncs}
\bibliography{predef,ref.bib,bib,shapiro-bib,local}

\begin{thebibliography}{10}

\bibitem{riak+crdt}
Cribbs, S., Brown, R.:
\newblock Data structures in {R}iak.
\newblock In: Riak Conference (RICON), San Francisco, CA, USA (oct 2012)

\bibitem{syn:optim:rep:1433}
Terry, D.B., Theimer, M.M., Petersen, K., Demers, A.J., Spreitzer, M.J.,
  Hauser, C.H.:
\newblock Managing update conflicts in {B}ayou, a weakly connected replicated
  storage system.
\newblock In: Symp.\ on Op.\ Sys.\ Principles (SOSP), Copper Mountain, CO, USA,
  ACM SIGOPS, ACM Press (December 1995)  172--182

\bibitem{app:rep:optim:1606}
DeCandia, G., Hastorun, D., Jampani, M., Kakulapati, G., Lakshman, A., Pilchin,
  A., Sivasubramanian, S., Vosshall, P., Vogels, W.:
\newblock {D}ynamo: {A}mazon's highly available key-value store.
\newblock In: Symp.\ on Op.\ Sys.\ Principles (SOSP). Volume~41 of Operating
  Systems Review., Stevenson, Washington, USA, Assoc.\ for Computing Machinery
  (October 2007)  205--220

\bibitem{rep:syn:sh138}
Shapiro, M., Pregui{\c c}a, N., Baquero, C., Zawirski, M.:
\newblock A comprehensive study of {C}onvergent and {C}ommutative {R}eplicated
  {D}ata {T}ypes.
\newblock Rapp.\ Rech. 7506, Institut National de la Recherche en Informatique
  et Automatique (INRIA), Rocquencourt, France (January 2011)

\bibitem{syn:rep:sh143}
Shapiro, M., Pregui{\c c}a, N., Baquero, C., Zawirski, M.:
\newblock Conflict-free replicated data types.
\newblock In D{\'e}fago, X., Petit, F., Villain, V., eds.: Int.\ Symp.\ on
  Stabilization, Safety, and Security of Distributed Systems (SSS). Volume 6976
  of Lecture Notes in Comp.\ Sc., Grenoble, France, {S}pringer-{V}erlag
  (October 2011)  386--400

\bibitem{alg:rep:sh132}
Letia, M., Pregui{\c c}a, N., Shapiro, M.:
\newblock {CRDTs}: Consistency without concurrency control.
\newblock Rapp.\ Rech. RR-6956, Institut National de la Recherche en
  Informatique et Automatique (INRIA), Rocquencourt, France (June 2009)

\bibitem{app:1639}
Baquero, C., Moura, F.:
\newblock Using structural characteristics for autonomous operation.
\newblock Operating Systems Review \textbf{33}(4) (1999)  90--96

\bibitem{idempotenceHelland:2012}
Helland, P.:
\newblock Idempotence is not a medical condition.
\newblock Queue \textbf{10}(4) (April 2012)  30:30--30:46

\bibitem{Brown:2014:RDM:2596631.2596633}
Brown, R., Cribbs, S., Meiklejohn, C., Elliott, S.:
\newblock Riak dt map: A composable, convergent replicated dictionary.
\newblock In: Proceedings of the First Workshop on Principles and Practice of
  Eventual Consistency. PaPEC '14, New York, NY, USA, ACM (2014)  1:1--1:1

\bibitem{rep:opt:sh151}
Bieniusa, A., Zawirski, M., Pregui{\c c}a, N., Shapiro, M., Baquero, C.,
  Balegas, V., Duarte, S.:
\newblock An optimized conflict-free replicated set.
\newblock Rapp.\ Rech. RR-8083, Institut National de la Recherche en
  Informatique et Automatique (INRIA), Rocquencourt, France (October 2012)

\bibitem{deltaCode}
Baquero, C.:
\newblock Delta-enabled-crdts.
\newblock Github Repo: https://github.com/CBaquero/delta-enabled-crdts

\bibitem{lattices:Brian:2002}
Davey, B.A., Priestley, H.A.:
\newblock Introduction to Lattices and Order (2. ed.).
\newblock Cambridge University Press (2002)

\bibitem{delta:almeida:2014}
Almeida, P.S., Shoker, A., Baquero, C.:
\newblock Efficient state-based crdts by delta-mutation.
\newblock CoRR \textbf{abs/1410.2803} (2014)

\bibitem{DBLP:conf/popl/BurckhardtGYZ14}
Burckhardt, S., Gotsman, A., Yang, H., Zawirski, M.:
\newblock Replicated data types: specification, verification, optimality.
\newblock In Jagannathan, S., Sewell, P., eds.: POPL, ACM (2014)  271--284

\bibitem{Parker:1983:DMI:1313337.1313753}
Parker, D.S., Popek, G.J., Rudisin, G., Stoughton, A., Walker, B.J., Walton,
  E., Chow, J.M., Edwards, D., Kiser, S., Kline, C.:
\newblock Detection of mutual inconsistency in distributed systems.
\newblock IEEE Trans. Softw. Eng. \textbf{9}(3) (May 1983)  240--247

\bibitem{MGS14}
Mukund, M., R., G.S., Suresh, S.P.:
\newblock Optimized or-sets without ordering constraints.
\newblock In: Proceedings ot the International Conference on Distributed
  Computing and Networking, New York, NY, USA, ACM (2014)  227–241

\bibitem{app:rep:optim:1501}
Wuu, G.T.J., Bernstein, A.J.:
\newblock Efficient solutions to the replicated log and dictionary problems.
\newblock In: Symp.\ on Principles of Dist.\ Comp.\ (PODC), Vancouver, BC,
  Canada (August 1984)  233--242

\bibitem{db:rep:optim:1454}
Johnson, P.R., Thomas, R.H.:
\newblock The maintenance of duplicate databases.
\newblock Internet Request for Comments RFC 677, Information Sciences Institute
  (January 1976)

\bibitem{conway2012logic}
Conway, N., Marczak, W.R., Alvaro, P., Hellerstein, J.M., Maier, D.:
\newblock Logic and lattices for distributed programming.
\newblock In: Proceedings of the Third ACM Symposium on Cloud Computing, ACM
  (2012) ~1

\bibitem{burckhardt2012cloud}
Burckhardt, S., F{\"a}hndrich, M., Leijen, D., Wood, B.P.:
\newblock Cloud types for eventual consistency.
\newblock In: ECOOP 2012--Object-Oriented Programming.
\newblock Springer (2012)  283--307

\bibitem{navalho2013incremental}
Navalho, D., Duarte, S., Pregui{\c{c}}a, N., Shapiro, M.:
\newblock Incremental stream processing using computational conflict-free
  replicated data types.
\newblock In: Proceedings of the 3rd International Workshop on Cloud Data and
  Platforms, ACM (2013)  31--36

\bibitem{deltaICDCS13}
Deftu, A., Griebsch, J.:
\newblock A scalable conflict-free replicated set data type.
\newblock In: Proceedings of the 2013 IEEE 33rd International Conference on
  Distributed Computing Systems. ICDCS '13, Washington, DC, USA, IEEE Computer
  Society (2013)  186--195

\bibitem{BAS2014}
Baquero, C., Almeida, P.S., Shoker, A.:
\newblock Making operation-based {CRDT}s operation-based.
\newblock In: Proceedings of Distributed Applications and Interoperable
  Systems: 14th IFIP WG 6.1 International Conference, Springer (2014)

\bibitem{export:211340}
Burckhardt, S., Leijen, D., Fahndrich, M.:
\newblock Cloud types: Robust abstractions for replicated shared state.
\newblock Technical Report MSR-TR-2014-43 (March 2014)

\bibitem{malkhi2007concise}
Malkhi, D., Terry, D.:
\newblock Concise version vectors in winfs.
\newblock Distributed Computing \textbf{20}(3) (2007)  209--219

\end{thebibliography}

\newpage
\appendix

\section{Proof of Proposition~\ref{prop:corr}}

\begin{proof}
By simulation, establishing a correspondence between an execution $E^\delta$,
and execution $E$ of a standard CRDT of which $(S,M^\delta,Q)$ is a decomposition,
as follows:
1) the state $(X_i, D_i, \ldots)$ of each node in $E^\delta$ containing
CRDT state $X_i$, information about delta-intervals $D_i$ and possibly other
information, corresponds to only $X_i$ component (in the same join-semilattice);
2) for each action which is a delta-mutation $m^\delta$ in $E^\delta$, $E$ executes
he corresponding mutation $m$, satisfying $m(X) = X \join m^\delta(X)$;
3) whenever $E^\delta$ contains a send action of a delta-interval $\dint iab$,
execution $E$ contains a send action containing the full state $X_i^b$;
4) whenever $E^\delta$ performs a join into some $X_i$ of a delta-interval
$\dint jab$, execution $E$ delivers and joins the corresponding message
containing the full CRDT state $X_j^b$.
By induction on the length of the trace, assume that for each replica $i$,
each node state $X_i$ in $E$ is equal to the corresponding component in the
node state in $E^\delta$, up to the last action in the global trace.
A send action does not change replica state, preserving the correspondence.
Replica states only change either by performing data-type update operations or
upon message delivery by merging deltas/states respectively. If the next
action is an update operation, the correspondence is preserved
due to the delta-state decomposition property $m(X) = X \join m^\delta(X)$.
If the next action is a message delivery at replica $i$, with a merging
of delta-interval/state from other replica $j$, because algorithm $A^\delta$
satisfies the causal merging-condition, it only joins into state $X_i^k$ a
delta-interval $\dint jab$ if $X_i^k \pgeq X_j^a$. In this case, the outcome
will be:
\begin{eqnarray*}
X_i^{k+1}
     &=& X_i^k \join \dint jab \\
     &=& X_i^k \join \bigjoin \{ d_j^l | a \leq l < b \}\\
     &=& X_i^k \join X_j^a \join \bigjoin \{ d_j^l | a \leq l < b \}\\
     &=& X_i^k \join X_j^a \join d_j^a \join d_j^{a+1} \join \ldots \join d_j^{b-1} \\
     &=& X_i^k \join X_j^{a+1} \join d_j^{a+1} \join \ldots \join d_j^{b-1} \\
     &=& \ldots \\
     &=& X_i^k \join X_j^{b-1} \join d_j^{b-1} \\
     &=& X_i^k \join X_j^b \\
\end{eqnarray*}
The resulting state $X_i^{k+1}$ in $E^\delta$ will be, therefore, the same as
the corresponding one in $E$ where the full CRDT state from $j$ has been joined,
preserving the correspondence between $E^\delta$ and $E$.
\end{proof}

\section{Proof of Proposition~\ref{prop:causal-algo}}

\begin{proof}
From Proposition~\ref{prop:causal-consistency}, it is enough to prove that
the algorithm satisfies the causal delta-merging condition.
The algorithm explicitly keeps deltas $d_i^k$ tagged with
increasing sequence numbers (even after a crash), according with the
definition; node $j$ only sends to $i$ a delta-interval $\dint jab$
if $i$ has acked $a$; this ack is sent only if $i$ has already joined some
delta-interval (possibly a full state) $\dint jka$. Either $k = 0$ or, by the
same reasoning, this $\dint jka$ could only have been joined at $i$ if some
other interval $\dint jlk$ had already been joined into $i$. This reasoning
can be recursed until a delta-interval starting from zero is reached. Therefore,
$X_i \pgeq \bigjoin \{ d_j^k | 0 \leq k < a \} = \dint j0a = X_j^a$.
\end{proof}

\end{document}